\documentclass{sigcomm-alternate}

\usepackage{paralist}
\usepackage{color}
\usepackage{subfigure}
\usepackage{balance}
\usepackage{url}
\usepackage{hyphenat}

\newcommand{\ndnname}[1]{{\footnotesize\path{#1}}}
\newcommand{\Adv}{Adv}

\newtheorem{theorem}{Theorem}[section]

\newtheorem{defn}[theorem]{Definition}

\begin{document} 
\title{Elements of Trust in Named-Data Networking}

\numberofauthors{2}
\author{
\alignauthor
Cesar Ghali \hspace*{1.5cm} Gene Tsudik \\
       \affaddr{University of California, Irvine\\
       \{cghali,gene.tsudik\}@uci.edu}
\alignauthor
Ersin Uzun\\
       \affaddr{Palo Alto Research Center \\
       Ersin.Uzun@parc.com}
}

\maketitle

\begin{abstract}
In contrast to today's IP-based host-oriented Internet architecture, 
Information-Centric Networking (ICN) emphasizes content by making it directly 
addressable and routable. Named Data Networking (NDN) architecture is an instance of 
ICN that is being developed as a candidate next-generation Internet architecture. By 
opportunistically caching content within the network (in routers), NDN 
appears to be well-suited for large-scale content distribution and for 
meeting the needs of increasingly mobile and bandwidth-hungry applications 
that dominate today's Internet.

One key feature of NDN is the requirement for each content object to be 
digitally signed by its producer. Thus, NDN should be, in principle, immune 
to distributing fake (aka ``poisoned'') content. However, in practice, this 
poses two challenges for detecting fake content in NDN routers: (1) overhead 
due to signature verification and certificate chain traversal, and (2) lack 
of trust context, i.e., determining which public keys are trusted to verify which 
content. Because of these issues, NDN does not force routers to verify 
content signatures, which makes the architecture susceptible to content poisoning 
attacks.

This paper explores root causes of, and some cures for, content poisoning 
attacks in NDN. In the process, it becomes apparent that meaningful mitigation 
of content poisoning is contingent upon a network-layer trust management architecture, 
elements of which we construct while carefully justifying specific design 
choices. This work represents the initial effort towards comprehensive trust 
management for NDN. 
\end{abstract}

\section{Introduction}
\label{sec:intro}
The Internet usage model has changed considerably over the last two decades. 
Limitations of the current Internet are becoming more pronounced as network 
services and applications become increasingly mobile and data-centric. In 
recent years, a number of research efforts have sprung up aiming to design 
the next-generation Internet architecture. Some are based on the notion of 
Information-Centric Networking (ICN) which emphasizes efficient and scalable 
content distribution. Named Data Networking (NDN)\cite{NDN}, a fork 
from PARC's Content Centric Networking (CCNx) architecture\cite{ccnx}, is one such research effort. 
One of the main tenets of NDN is named content. 
NDN also stipulates in-network content caching, by routers. To secure 
each content, NDN requires it to be cryptographically signed by its producer. 
This way, globally addressable and routable content can be authenticated by anyone, 
which allows NDN to decouple trust in content from trust in entities that store 
and disseminate it. NDN entities that request content are called {\em consumers}. 
A consumer is expected to verify content signatures in order to assert:
\begin{compactitem}
\item {\em Integrity} -- a valid signature (computed over a content hash) 
guarantees that signed content is intact;
\item {\em Origin Authentication} -- since a signature is bound to the public 
key of the signer, anyone can verify whether content originates with its claimed 
producer;
\item {\em Correctness} -- since a signature binds content name to its payload, 
a consumer can securely determine whether delivered content corresponds to what 
was requested;
\end{compactitem}
Although any NDN entity can verify any content signature, NDN routers are 
{\em not required}, to do so. This is not only because of the overhead 
stemming from the actual cryptographic verification of the signature itself. 
There are two other, more important, reasons for not mandating router verification 
of content signatures: 
\begin{compactenum}
\item First, a router must be aware of the specific trust model for each 
content-producing application. Given the wide range of possible applications,
it is very unlikely that they will all use the same trust model.
Some applications will probably use trust hierarchies, while others might adopt
a flat peer-based trust models, or hybrid versions thereof. Furthermore, the set
of NDN applications will change over time. Also, the trust model of a particular 
application might not be static in the long term. 
\item Second, depending on the trust model of an application 
associated with a particular content, a router needs access to  -- and thus 
might need to fetch\footnote{The alternative of carrying the entire collection of 
certificates as part of each content is clearly undesirable.}  -- 
multiple public key certificates or similar structures in order to trust the public key 
that verifies a content signature. 
For example, if an application uses a hierarchical PKI, an
entire root-to-leaf path might have to be traversed and all intermediate 
certificates would need to be separately verified. This would need to 
include ancillary activities for each such certificate, i.e.,  expiration 
and revocation checking. 
\end{compactenum}
These issues greatly complicate network-layer trust management in NDN. One easy 
alternative -- adopted by the current version of NDN -- is to make it 
optional for routers to verify content signatures. Unfortunately, this 
decision leaves NDN vulnerable to content poisoning attacks on router 
caches. To make matters worse, NDN does not provide any definitive mechanism 
for a consumer to request genuine desired content. Instead, a consumer 
that receives fake content can explicitly exclude the latter (by referring 
to its hash) in subsequent requests. This does not guarantee eventual 
success, due to the potentially unbounded number of fake content objects 
sharing the same name.

This undesirable state-of-affairs serves as the main motivation for our 
work. In this paper, we analyze NDN architecture and its susceptibility 
to content poisoning attacks. Next, we postulate some intuitive goals for 
routers to support trust management and content validation. We then
present simple rules that allow NDN parties (consumers, producers and
routers) to mitigate content poisoning, while minimizing trust-related 
complexity for routers. These rules require no changes to the fundamentals 
of the NDN architecture.

Besides being the first effort to address content poisoning and trust 
management in NDN, one contribution of this work is in careful analysis 
and justifications for placement and complexity of various trust mechanisms.

\noindent{\textbf{Disclaimer:}}
it is impossible to predict whether NDN
will ever cross the line between a research prototype and a widely deployed architecture.
NDN, similar to every other candidate Future Internet Architecture, has its benefits and pitfalls.
The purpose of this work is not to advocate for or against NDN.
Instead, we aim to improve NDN security features by 
utilizing techniques 
in~\cite{smetters2009securing,ghodsi2011naming,ghodsi2011information,fayazbakhsh2013less, koponen2007data}. 
Our main goal is to provide NDN routers with a mechanism 
to efficiently and securely verify content in order to mitigate content poisoning attacks.

\noindent{\textbf{Scope:}} as reflected in the title, this paper focuses 
on network-layer trust issues, motivated by the content poisoning problem. 
We do not address other NDN security issues, such as interest flooding attacks 
\cite{compagno2013poseidon, gasti2012ddos, afanasyev2013interest}, cache 
pollution attacks \cite{conti2013lightweight} and routing security~\cite{yi2013case}.

\section{NDN Overview}
\label{sec:ndn}
Unlike IP which focuses on end-points of communication and their names/addresses, 
NDN (\cite{jacobson2009networking,NDN}) emphasizes content and makes it named, 
addressable and routable at the network layer. A content name is composed of 
one or more variable-length components opaque to the network. Component 
boundaries are explicitly delimited by ``\ndnname{/}'' in the usual path-like 
representation. For example, the name of CNN 
home-page content for August 20, 2014 might be: 
\ndnname{/ndn/cnn/news/2014august20/index.htm}. Large content can be split into 
segments with different names, e.g., fragment 37 of Alice's YouTube video 
could be named:\\
\ndnname{/ndn/youtube/alice/video-749.avi/37}. 

NDN communication adheres to the {\em pull} model and content is delivered 
to consumers only following an explicit request. There are two types of 
packets in NDN: interest and content. A consumer requests content by issuing 
an {\em interest} packet. If an entity can ``satisfy'' a given interest, 
it returns a corresponding {\em content} packet. Content delivery must be 
preceded by an interest. If content $C$ with name $n$ is received by a 
router with no pending interest for $n$, $C$ is considered unsolicited and 
is discarded. Name matching in NDN is prefix-based. For example, an interest 
for \ndnname{/ndn/youtube/alice/video-749.avi} can be satisfied by content named 
\ndnname{/ndn/youtube/alice/video-749.avi/37}.\footnote{The reverse 
does not hold, by design.}

NDN content includes several fields. In this paper, we are only interested 
in the following:
\begin{compactitem}
\item \texttt{Signature} -- a public key signature, generated by the content 
producer, covering the entire content, including all explicit components of 
the name and a reference to the public key needed to verify it.
\item \texttt{Name} -- a sequence of explicit name components followed by an 
implicit \texttt{digest} (cryptographic hash) component of the content that is recomputed at every 
hop. This effectively provides each content with a unique name and guarantees 
a match with a name provided in an interest. However, in most cases, the \texttt{digest} 
component is not present in interest packets, since NDN does not provide any 
secure mechanism for a consumer to learn a content hash, prior to requesting it. 
\item \texttt{PublisherPublicKeyDigest} (PPKD) -- an SHA-256 digest of the public 
key needed to verify the content signature.
\item \texttt{Type} -- content type, e.g., data, encrypted content, key, etc.
\item \texttt{Freshness} -- recommended content lifetime (after being cached)
set by the producer.
\item \texttt{KeyLocator} -- a reference to the public key required to verify 
the signature. This field has three options: (1) verification key, (2) certificate 
containing the verification key, or (3) NDN name referencing the content that 
contains the verification key.
\end{compactitem}
Each content producer must have at least one public key, represented 
as a {\em bona fide} named content of \texttt{Type}$=key$, signed by its issuer, 
e.g., a certification authority (CA).\footnote{Recall that NDN is agnostic as 
far as trust management, aiming to accommodate peer-based, hierarchical and 
hybrid PKI approaches.} The naming convention for a public key content object is to contain 
``key'' as its last explicit component, e.g., 
\ndnname{/ndn/russia/moscow-airport/transit/snowden/key}. 

An NDN interest includes the following fields:
\begin{compactitem}
\item \texttt{Name} -- NDN name of requested content.

\item \texttt{Exclude} -- contains information about name components 
that {\bf must not} occur in the name of returned content. This field 
can also be used to exclude certain content by referring to its \texttt{digest}, 
which, as noted above, is included in the content as an implicit last component 
of each content name, or in a separate field.
\item \texttt{PublisherPublicKeyDigest (PPKD)} -- the SHA-256 digest of the 
publisher public key. If this field is present in the {\em interest}, 
a matching content objects must have the same digest in its 
\texttt{PPKD}.
\end{compactitem}
There are three types of NDN entities\footnote{Note that a physical 
entity (a host, in today's parlance) can be both consumer and producer 
of content.}: (1) {\em consumer} -- an entity that issues an interest 
for content, (2) {\em producer} -- an entity that produces and publishes 
(as well as signs) content, and (3) {\em router} -- an entity that 
routes interest packets and forwards corresponding content packets. 
Each entity (not just routers) maintains the following three 
components:
\begin{compactitem}
\item {\em Content Store} (CS) -- cache used for content caching and 
retrieval. From here on, we use the terms {\em CS} and {\em cache} 
interchangeably.
\item {\em Forwarding Interest Base} (FIB) -- routing table of name 
prefixes and corresponding outgoing interfaces used to route interests. 
NDN does not specify or mandate any routing protocol. Forwarding is 
done via longest-prefix match on names.
\item {\em Pending Interest Table} (PIT) -- table of outstanding 
(pending) interests and a set of corresponding incoming and outgoing 
interfaces.
\end{compactitem}
When a router receives an interest for content named $n$ which is 
not in its cache, and there are no pending interests for the same 
name in its PIT, it forwards the interest to the next hop(s), according 
to its FIB. For each forwarded interest, a router stores some amount of 
state information, including the name in the interest and the interface 
on which it arrived. However, if an interest for $n$ arrives while 
there is a pending entry for the same content name in the PIT, the 
router collapses the present interest (and any subsequent interests 
for $n$) storing only the interface upon which it was received. If and 
when content is returned, the router forwards it out on all 
incoming-interest interfaces and flushes the corresponding PIT entry. 
Since no additional information is needed to deliver content, an 
interest does not carry any {\em source address}. (If a content 
fails to arrive before some router-determined expiration time, the 
router can either flush the PIT entry or attempt interest retransmission 
over the same or different interfaces.)

An NDN router's cache size is determined by local resource availability. 
Each router unilaterally determines which content to cache and for 
how long, though lifetime (as mentioned above) can be recommended by the 
producer. Upon receiving an interest, a router first checks its 
cache to see if it already has requested content in its cache. 
Producer-originated content signatures allow 
consumers and routers to authenticate received content, regardless 
of the entity serving it. 

\section{Content Poisoning}
The central objective of NDN is efficient and scalable distribution of 
information. This is facilitated by routers opportunistically caching 
content. Whenever an NDN router receives an interest for a name that 
matches a content in its cache, it satisfies the interest with that 
content. Since routers are not required to verify signatures, the 
delivered content is not guaranteed to be authentic. However, a consumer 
is required to verify signatures of all returned content. A consumer 
is thus assumed to have the necessary application-specific trust context 
to decide which public keys to trust. This allows consumers to reliably 
detect fake content. 

However, NDN offers no means for consumers to ask routers to {\em flush} 
fake content from their caches. The only recourse for a consumer 
that detects fake content is to issue another interest that specifically 
excludes the unwanted content by specifying its hash in the exclusion 
filter field of the new interest. Unfortunately, this explicit exclusion 
does not signify (to routers) bad or poisoned content, as the same feature can
also be used to exclude stale content. Furthermore, even if the 
exclusion technique were to be used strictly for flagging poisoned 
content, the result would be undesirable, for the following reasons:

The entire notion of consumers (i.e, end-systems or hosts) informing 
routers about poisoned content is full of pitfalls. Suppose a consumer 
complains to a router about specific content. If this is done without 
consumer authentication (whether via an interest, e.g, using exclusion, 
or via a separate packet type), the router would have two choices: (1) 
immediately flush referenced content from its cache, or (2) verify the 
content signature and flush content only if verification fails. The 
former (1) is problematic, since it opens the door for anyone to cause 
easy removal of popular content from router caches, which can be 
considered as a type of a denial-of-service attack. Even if this were 
not an issue, there would remain a more general problem: as noted in 
\cite{gasti2012ddos}, the adversary mounting a content poisoning attack 
could continue {\em ad infinitum} to feed new invalid content in 
response to interests that exclude previously consumer-detected invalid 
content. The second option (2) is also problematic, because, besides 
the cost of verifying a signature (which can lead to a denial-of-service 
attack by itself), it brings back the problem of routers having to understand 
potentially complex trust semantics of many diverse content-producing 
applications.

Another possibility is to require consumers to authenticate themselves 
when complaining about poisoned content. This would entail signing the 
interest (or another new packet type) that complains about allegedly bad 
content. One unpleasant privacy consequence is that the signer (consumer) 
would be exposed by the signature, since it would need to be bound to a 
public key, contained in a certificate. (This certificate would have to 
be communicated with each complaint message, along with auxiliary 
information that the router would need to trust the certificate.) 
More generally, signing would violate one of the key elements of NDN 
architecture -- consumer opacity. Recall that producers sign content,
while consumers do not sign interests, or any other messages. 

Another reason why consumer signing of ``complaint'' messages is 
problematic is because it can be abused to mount DoS attacks on routers by 
flooding them with junk complaints and forcing expensive signature 
verification.\footnote{The same attack does not work with 
flooding of routers with junk content since content can not be sent unsolicited 
and a router would only attempt 
signature verification of incoming content for which it has a pending 
interest entry in its PIT.} Note that, even if the router successfully 
authenticates a consumer complaint, this is no guarantee that the accused 
content is fake; in order to be sure, the router would have to verify the 
content signature as well. Moreover, authentication of consumers by 
routers would require identity management and verification systems to 
be in place at the network layer, thus adding significant overhead.

Finally, the preceding discussion applies not only to content {\em cached} 
by routers. Since NDN only recommends, and does not mandate, content 
caching, it is entirely {\em legal} for a router not to cache some, or 
all, content that it forwards. If a router does not cache $C$, then complaining 
about $C$ being fake is clearly useless. 

At this point, it becomes clear that dealing with fake content represents 
a real challenge for NDN. Although some light-weight non-cryptographic 
and partially effective counter-measures have been proposed (e.g., \cite{ghali2014needle}), 
they do not fully address the problem and quickly become ineffective against an active adversary.

\subsection{Zooming In}
Based on the above arguments and recent results simulating content-poisoning attacks 
\cite{ghali2014needle}, we conclude that NDN has a major security problem, since it offers: (1) 
no way to prevent fake content from being delivered to consumers, and (2) no way to 
reliably flush invalid content from router caches. There are two reasons for this 
problem:

%
%
\noindent {\bf 1. Ambiguous interests}: NDN requires each interest to carry the name of 
desired content. However, neither the \texttt{digest} component of the name, nor the 
\texttt{PPKD} is a required field in an interest. In other words, 
an interest for a content name can be satisfied by multiple content objects, 
including those with untrusted or unverifiable signatures. 

\noindent {\bf 2. No unified trust model}: even if routers could verify signatures at line 
speed, NDN does not provide a trust model enforceable at the network layer. Although 
two aforementioned selector fields can be used to communicate content-specific trust 
context to the network layer, NDN has no mechanism for a consumer to securely pre-acquire 
the hash of a given content, or the specific public key that should be used to verify a 
content signature. 

%
%
In order to demonstrate the grave effect content poisoning can have on NDN, we conducted 
a simple experiment using ndnSIM~\cite{afanasyev2012ndnsim} -- a simplified implementation 
of NDN architecture as a NS-3~\cite{ns3} module. Our results verified that content 
poisoning can significantly delay or block customers from accessing valid content. 
Details about the experiment setup and results can be found in Appendix A.

\subsection{Goals}
As a first step in addressing the content poisoning problem, it is necessary to acknowledge 
the obvious, i.e., that {\bf network\hyp{}layer trust management and content poisoning 
are inseparably conjoined}. Since content is the basic unit of network\hyp{}layer 
``currency'' in NDN, trust in content (and not in its producers or consumers) is the 
central issue at the network layer.

Second, trust-related complexity (activities, state maintenance, etc.) must be minimized 
at the network layer. Specifically, as part of validating content, {\bf a 
router should not: fetch public key certificates, perform expiration and revocation 
checking of certificates, maintain its own collection of certificates, or be aware of 
trust semantics of various applications}.\footnote{This is separate from trust 
management  for routing protocols.}

On a related note, we claim that, ideally, {\bf a router should verify at most one signature per content}. 
This upper-bounds the heavier part of content-related cryptographic overhead; the other part 
is computing a content hash. Ideally, a router would not perform any signature 
verification at all. However, as discussed below, this might be possible for some, 
yet not all, content. Also, although verifying a signature given an appropriate 
public key is a mechanical operation, a router would still need to support multiple 
signature algorithms since uniformity across all applications is improbable.

The above discussion implies that NDN entities other than routers, i.e., {\bf producers 
and consumers of content, should bear the brunt of trust management}. 

\section{The Interest-Key Binding Rule}
\label{sec:ikb}
Ghodsi et al.~\cite{ghodsi2011naming} informally argue that, for each content, at least two out 
of three possible bindings (producer-key, name-key, producer-name) must be
present. The third binding is transitively inherited from the other two. Due to the use of
human-readable names in NDN, producer--name binding can be easily inferred.\footnote{If we assume
that names are clear and unambiguous.} Our approach to network-layer trust adheres
to all goals outlined above. It is based on the binding between a name and the 
public key used to verify the content signature. We denote it as the {\em 
Interest-Key Binding (IKB)} rule:

\vspace{0.2cm}
\centerline{\fbox{\parbox{0.968\columnwidth}{\textcolor{blue}{\sf {\bf IKB:}
An interest must reflect the public key of the producer.}}}}
\vspace{0.2cm}

A very similar concept --  {\em self-certifying naming scheme} -- is described 
in \cite{ghodsi2011information}. As discussed in Section~\ref{sec:optimizations},
this concept needs to be adjusted for the NDN context.
Recall that NDN interest format (Section \ref{sec:ndn}) includes an optional field 
\texttt{PPKD} which serves exactly this purpose. Our approach 
makes it mandatory without any substantive changes to the NDN architecture.

An NDN public key is a special type of content in the form of a certificate signed by 
the issuing CA. Each certificate contains a list of all name prefixes that it is 
authorized to sign/verify. The name of the certificate-issuing (content-signing) CA 
and the name of the key contained in a certificate (content) are not required to have 
any specific relationship. This is part and parcel of NDN's philosophy of leaving trust management 
up to the application, e.g., signed content $C$ can be verified with public key $PK$ 
with $C$ and $PK$ having no common prefix. For instance, content containing 
the public key \ndnname{/cnn/usa/web/key} could be issued and verified by the key 
\ndnname{/verisign/key}. Of course, an application is free to impose all kinds of 
restrictions, as long as routers remain oblivious.

\subsection{Implications for Producers and Routers}
\label{sec:imp-pr}
We now examine IKB implications on content producers and routers, respectively.

For content producers, IKB has very few consequences. In fact, it simplifies 
content construction by asking the producer to include the public key itself in the 
\texttt{KeyLocator} field of content. In other words, IKB obviates two other current 
NDN options: (1) referring to a verification key (via the \texttt{KeyLocator} 
field) by its name, or (2) including it in a form of a certificate.

For NDN routers, IKB implications are overwhelmingly positive. 
First, a router needs to perform no fetching, storing or parsing 
of public key certificates, as well as no revocation or expiration checking. 
All such activities are left to consumers.

Upon receiving a content and identifying the PIT entry (corresponding 
to one or more pending interests) a router simply hashes the public key from the 
content \texttt{KeyLocator} field and checks whether it matches the 
\texttt{PPKD} of the PIT entry. In case of a mismatch, 
the content is discarded.\footnote{A slightly simpler alternative is to perform PIT lookup 
each incoming content by using both content name and public 
key hash.} Otherwise, the content signature is verified and (if valid)
the content is forwarded and cached.

The implications would be even more beneficial for producers and routers with the use of
self-certifying content names (SCNs), as discussed in Section \ref{sec:optimizations} below. 
With this optimization, inclusion of key information and signature checking could be 
avoided for most content objects, thus further reducing the communication and computation overhead.

\subsection{Implications for Consumers}
\label{sec:imp-c}
For consumers, IKB does not increase complexity. It actually prompts us to codify 
desired consumer behavior -- something that has been left unspecified in the NDN 
architecture.

The most immediate IKB consequence for a consumer is the need to {\bf obtain and 
validate the producer's public key before issuing an interest for any content 
originated by that producer.} At the first glance, this might appear to be an 
example of the proverbial ``chicken-and-egg'' problem. However, we show below that 
this is not the case.

A consumer that wants to fetch certain content $C$ is doing so as part of some 
NDN application, $APP_C$. We assume that a consumer must have already installed 
this application. $APP_C$ must have a well-defined trust management architecture 
that is handled by its consumer-side software. However, the remaining question is: 
how to bootstrap trust and how to obtain initial public keys?

We consider three non-exclusive alternatives:

\noindent {\bf (1)} One possibility is that $APP_C$ client-side software comes 
with some pre-installed root public key(s), perhaps contained within self-signed 
certificates. Without loss of generality, we assume that there is only one such 
key -- $PK_{root}$. Armed with it, a consumer can request lower-level certificates, 
by issuing an interest referencing the hash of $PK_{root}$ in the 
\texttt{PPKD} field.\footnote{If $APP_C$ comes with several 
root public keys, the consumer would need to issue multiple simultaneous interests referencing 
the hash of each root key in \texttt{PPKD}.}

\noindent {\bf (2)} Alternatively, one could imagine a global Key Name Service (KNS), somewhat akin 
to today's Domain Name Service (DNS). In response to consumer-issued interests 
referencing public key names and/or name prefixes, KNS would reply with signed 
content containing one or more public key certificates (i.e., as embedded content) 
corresponding to requested names.

\noindent {\bf (3)} A similar approach is a global search-based service, i.e., 
something resembling today's Google. A consumer would issue a search query (via an 
interest) to the search engine which would reply with signed content representing a 
set (e.g., one page at a time) of query results. One or more of those results would 
point to content corresponding to the public key certificate of interest to the consumer.

In cases (2) and (3), consumers would still need to somehow securely obtain the root 
public keys for KNS and the search engine, respectively. This can be easily done via (1).

\subsection{Security Arguments}
We now return to the original motivation for this work -- mitigation of content 
poisoning attacks. We need to show that global adherence to the IKB rule leads 
to security against content poisoning.\footnote{Note that, as mentioned in Section 
\ref{sec:intro}, cache pollution attacks are an entirely different matter.}

If we assume that:
\begin{compactenum}
\item Every router abides by the IKB rule and acts as described in Section \ref{sec:imp-pr}.
\item Every consumer abides by the IKB rule and acts as described in Section \ref{sec:imp-c}.
\item The consumer requesting content $C$ is not malicious.
\item Each router $R$ that is one hop away from the consumer is not compromised.
\item The links between a consumer and its adjacent routers are not compromised.
\end{compactenum}
We can briefly argue security by contradiction: Suppose that a consumer receives 
fake content $C$ from $R$. Let $\mathit{Int}$ denote the interest (issued 
earlier by that consumer) that was satisfied by $C$. According to IKB, $\mathit{Int}$ 
must contain the digest of a public key of producer $P$ in its \texttt{PPKD} 
field. Let $PK$ denote this public key. Consequently, $R$ must have made sure that: 
(1) $C$ is signed with a public key $PK'$ with a hash matching \texttt{PPKD} 
of $\mathit{Int}$, meaning that $H(PK')=H(PK)$ and (2) the signature itself is 
correct, i.e., valid. Also, since $R$ is not malicious and all communication between 
$R$ and (also not malicious) the consumer is secure, the only remaining possibility 
is a hash collision, i.e., $PK'\neq~PK$ while $H(PK')=H(PK)$. The latter is assumed 
to occur with negligible probability.

This does not yet conclude our security discussion. As noted in \cite{gasti2012ddos}, content poisoning 
attacks can originate with malicious routers. What happens if a malicious router $R'$ feeds
poisoned content $C'$ to its non-malicious next hop neighbor $R$, towards some consumer(s)?
Since $R$ is honest and implements IKB, before forwarding and (optionally) caching
$C'$, it verifies, as before, that the signature of $C'$ is successfully verifiable using $PK$
that matches the hash in the corresponding PIT entry, i.e., the value of the \texttt{PPKD} field
of the original interest $\mathit{Int}$ that triggered creation of this PIT entry.

A more detailed security argument is provided in Appendix B.

\section{Optimizations}\label{sec:optimizations}

As mentioned before, IKB rule implies that routers should perform only one signature 
verification using the public key provided (by the producer) in the content and specified 
(by the consumer) using the \texttt{PPKD} field in the interest. 
Instead of including the public key in the content, it could be directly included by 
the consumer in the interest. This would require storing the public key alongside the interest 
in the PIT entry, to be used later for signature verification of the content. Since it 
is fair to assume that cache entries have longer lifetime than PIT entries, this approach 
can be beneficial in terms of storage. Its main drawback, however, is that the current 
interest format would need to be modified to include public keys.

For backbone routers that process and forward tens of gigabits per second, performing 
even a single signature verification per packet imposes a huge overhead. One approach 
to overcome this problem is to take advantage of the network structure. The current 
Internet is divided into Autonomous Systems (AS-s), each representing an administrative 
entity. In this architecture, only border routers of consumer-facing AS-s might implement 
the IKB rule by verifying signatures of all received contents. Alternatively, each router 
in an AS might probabilistically verify signatures on a subset of packets it forwards. The 
drawback of these approaches are that fake content could still be cached by routers that 
did not verify its signature. However, either method would have good chance of detecting 
and discarding most fake content before reaching to the consumer.

Another way to reduce signature verification overhead is to use SCNs 
\cite{ghodsi2011information,fu2000fast,mazieres1999separating,gasti2012ddos,baugher2012self}. 
According to \cite{wilcox2003names}, a content name can only have at most two out of the 
following three properties: security, uniqueness and human-readability. As suggested in \cite{ghodsi2011naming, 
ghodsi2011information}, SCNs can be formed by appending to the producer's public key digest a 
label that uniquely identifies the content. While this approach guarantees security and uniqueness, it 
lacks the means of verifying the binding between the content and its name~\cite{smetters2009securing}. 
To overcome this issue in NDN, we consider forming an SCN by specifying
the hash of requested content as the last component of the content name in the 
interest~\cite{jacobson2013named}. This provides name-content as well as producer-name, 
producer-key and name-key security bindings (as in Section~\ref{sec:ikb}). Although this use of SCNs 
does not yield fully human-readable names, it provides uniqueness and security 
properties~\cite{ghodsi2011naming}.

If a benign NDN consumer uses SCNs, the network guarantees (due to longest-prefix matching) 
delivery of ``valid'' content. The main advantage is that routers no longer 
need to verify signatures. Instead, they only recompute a content hash and check that 
it matches the one in the corresponding PIT entry. The remaining question is: how can a 
consumer obtain the hash of a content beforehand?

For the type of communication where most content is requested using SCNs, we advocate the 
use of so-called {\em catalogs}. A catalog is basically an 
authenticated data structure that includes a set of SCNs. This 
set can consists of references to content objects containing data, public keys, or  
other catalogs. The structure of catalogs can be application-specific and might vary 
from a simple list of SCNs, to multiple SCN sets forming a Merkle tree~\cite{merkle1982method}
or some similar data structure. To securely fetch an initial catalog, 
a consumer can fall back to using the \texttt{PPKD} interest field, as 
discussed earlier.

One obvious corollary of using SCNs in interest messages is that consumers 
and routers are no longer required to verify content signatures, as long as the 
SCN is trusted, i.e., obtained from a (consumer-verified) catalog. This reduces: (1) 
overhead of publishing, since producers now sign catalogs rather than individual content, and (2) 
network overhead, since there is no need to add the public key to the \texttt{KeyLocator} 
field of the content, as discussed in Section~\ref{sec:ikb}. The only time a signature  
is required is whenever a content is requested via \texttt{PPKD} 
interest field. In that case, both routers (prior to serving content from cache or forwarding it) 
and consumers (prior to accepting) {\em must} perform content signature verification. We believe
that is should be left up to the producer to decide whether a content should be requested by specifying its 
corresponding public key, SCN, or both.

Using SCNs in conjunction with catalogs brings up the issue of unsigned content objects. In other words,
a content $C$ which is indirectly signed as part of a catalog, can be fetched by its SCN, i.e., name-hash
combination. This does not rule out $C$ being separately signed by its producer. However, signing
a catalog-ed content increases overhead for the producer and increases content size. A sensible approach
is not to sign catalog-ed content  objects at all. This would imply that such objects can {\em only} be 
fetched via SCN. However, NDN architecture requires each content to be {\em individually} verifiable. 
Thus, existence of unsigned objects conflicts with a basic tenet of NDN \footnote{This is not the case for the 
latest version of CCNx, the original architecture that spawned NDN. CCNx 1.0 adopts secure catalogs 
(called manifests) and its packet format supports unsigned content objects\cite{ccnx}.}.

\section{Proposed Model in Practice}
\label{sec:trust_practice}
NDN was designed as a candidate next-generation Internet architecture.
In order to provide a smooth and successful transition path, NDN must 
contend with application-specific requirements, such as trust. In this section we discuss 
how the aforementioned trust model and its optimization could be applied in practice. 
We start by identifying different traffic types.

\subsection{Content Distribution}
This type of traffic corresponds to client-server communication in and accounts for well over 90\% of 
current Internet traffic\cite{GIP}. Since most requested content 
is static, creating secure catalogs is straightforward. Consumers request 
catalogs and then use included SCNs to request desired content.
We consider two common examples of content distribution traffic:

\vspace{.2cm}\noindent{\underline{Audio/Video Streaming:}} 
A typical audio/video is a large content split into several 
segments with different names (as mentioned in Section~\ref{sec:ndn}). If a catalog 
containing the SCNs of all the segments can be provided, consumers 
can use these names in subsequent interests to retrieve all segments of the content.

\vspace{.2cm}\noindent{\underline {Internet Browsing:}} 
We anticipate that most HTML files would fit into a single 
content object ~\cite{ramachandran2010web, everts2013the}. A typical HTML file contains reference 
links to other static and dynamic content, such as images, audio or other HTML pages 
(sub-pages). While rendering HTML files, Internet browsers parse all reference links 
and download corresponding content. Therefore, if an HTML file uses SCNs as references, 
it can be viewed and treated as a secure catalog. Of course,
SCNs can only be used for static content, since the hash of 
dynamic (e.g., generated upon request) content cannot be known {\em a priori}.

Internet browsing provides a good example of content that can be requested via either
\texttt{PPKD} or SCNs. Suppose that
a web page $A$ contains a reference link to sub-page $B$ and this link is expressed using 
an SCN. Once a consumer requests and obtains page $A$, the client browser
can request $B$ using the appropriate SCN in $A$. Whereas, other consumers 
might wish to directly request page $B$ (not as part of $A$) using its PPKD.
Note that, for obvious reasons, SCNs can not be used with HTML pages (or any other content)
with circular references, e.g., $A$ $\leftrightarrow$ $B$.

\subsection{Interactive Traffic}
\balance
Another major traffic type corresponds to interactive communication, where content is generated 
on demand. Applications such as voice/video conferencing, remote terminal access  
and on-line gaming fall into this category.  Such applications generally benefit from network caching 
only in cases of packet loss, since re-issued interests for lost packets are likely to be 
satisfied by the first hop NDN router.  
Obviously, the use of large catalogs for interactive real-time traffic is neither sensible nor feasible. 
Instead, consumers should request content by using 
\texttt{PPKD} in interests, in conjunction with producers perhaps 
offering small dynamically-generated catalogs, if short delays can be tolerated.

\section{Related Work}
Some prior research efforts discussed naming in content-oriented networks and 
its relationship to security. Notably, \cite{ghodsi2011naming} proposes establishing bindings 
between three ICN entities : (1) real-world identity 
coupled with the the producer of each content object, (2) name, and (3) public 
key used to verify the object signature. Only two of the three possible bindings (real-world
identity--name, name--key and real-world identify--key) are required, while the third can be 
transitively inherited. However, it is unclear how these bindings can be practically 
applied in the specific NDN settings.

Self-certifying naming schemes are discussed in~\cite{ghodsi2011naming, ghodsi2011information, 
fayazbakhsh2013less, koponen2007data}. Names are of the form $P:L$ where $P$ is 
the digest (hash) of the producer's public key, and $L$ is a label set by the producer. It is the 
latter's responsibility to make sure that names of this form are unique. This 
guarantees the name--key binding and trades off  human readability of names for 
strong security properties. Although, NDN use human-readable names, name-key binding 
is achievable by adding the \texttt{PPKD} field to interest messages. This allows 
interest forwarding based on longest-prefix matching on names. Whereas, using the $P:L$ scheme 
in NDN would result in tremendously large routing tables. We again recall that self-certifying 
names in NDN~\cite{jacobson2013named} are composed by adding the hash 
of the content as a name suffix (last component).

Prior work on Denial of Service (DoS) attacks on NDN 
includes \cite{compagno2013poseidon} and \cite{afanasyev2013interest}. Both results addressed 
a specific DoS attack type -- Interest Flooding -- based on inundating routers with spurious interest 
messages. Content poisoning was identified  in~\cite{gasti2012ddos}, which also sketched out 
some tentative countermeasures. Subsequently, \cite{ghali2014needle} proposed the first concrete
(however, only probabilistic) countermeasure based on analyzing exclusion patterns for cached content.

Trust and trust management systems are well studied in the 
literature, especially, in distributed environments, such as MANETs, ad hoc and wireless sensor 
networks (WSNs). \cite{cho2011survey} surveys the state of the art in trust 
management systems for MANETs. It emphasizes the need to combine the notions 
of ``social trust'' with ``quality-of-service (QoS) trust''. A similar 
survey can be found in~\cite{omar2012certification}.
\cite{lopez2010trust} presents an extensive review of trust management systems in WSNs. 
Based on unique features of WSNs, trust management system's best practices are derived and
state of the art countermeasures are evaluated against them.  
\cite{zahariadis2010trust} discusses security challenges in designing WSNs. It
distinguishes between the definitions of trust and security, and shows that cryptography is not 
always the solution for trust management. Instead, techniques from other domains should be 
included in defining and formalizing trust.

Since a single trust metric might not suffice to express trustworthiness of nodes, 
a multi-dimensional trust management framework is suggested in~\cite{li2010coping}. 
Three metrics are used: (1) node collaboration to perform tasks, such as packet 
forwarding, (2) node behavior, e.g., flagging nodes that flood the network, and (3) 
correctness of node-disseminated information, e.g., routing updates.

\cite{conner2009trust} proposes a framework for calculating a network entity's reputation 
score based on previous interactions feedback. In this framework, each service can
apply its own reputation scoring functions. It also supports caching of trust 
evaluation to reduce network overhead, and provides an API for reporting feedback 
and calculating reputation scores.

Policymaker \cite{blaze1996decentralized} is a tool that provides privacy and 
authenticity for network services. It offers a flexible and unified language for expressing 
policies and relationships. It also includes a local (per site or network) engine for carrying all trust 
operations, such as granting access to services.

All aforementioned techniques involve keeping track of other nodes' behavior in order to decide whether 
they are trusted. However, this general strategy is a poor match for NDN, since routers need an efficient 
mechanism to trust content, and not other entities. Because content can be served from anywhere
it is impractical for routers to trust other entities.

\section{Conclusion}
As argued in this paper, the NDN architecture is inherently susceptible to content poisoning 
attacks. To mitigate these attacks, we postulated some intuitive trust management goals needed to 
support content validation in NDN routers. We then presented 
simple rules that allow all NDN entities to validate content. These 
rules are compatible with the tenets of the NDN architecture. We also suggested  
several optimization techniques.

\subsection*{Acknowledgments}
This work was funded by NSF under the ``CNS-1040802: FIA: Collaborative Research: 
Named Data Networking (NDN)'' project. The authors thank Paolo Gasti, 
Christopher Wood, CCR reviewers and several anonymous commenters for their feedback 
on earlier drafts of this paper.

\bibliographystyle{abbrv}
\bibliography{IEEEabrv,references}
\normalsize


\begin{figure*}[t]
\centering
\subfigure[DFN topology - each edge router above is connected to 5 NDN
consumers]
{
	\fbox{\includegraphics[height=2.3in,width=0.86\columnwidth]{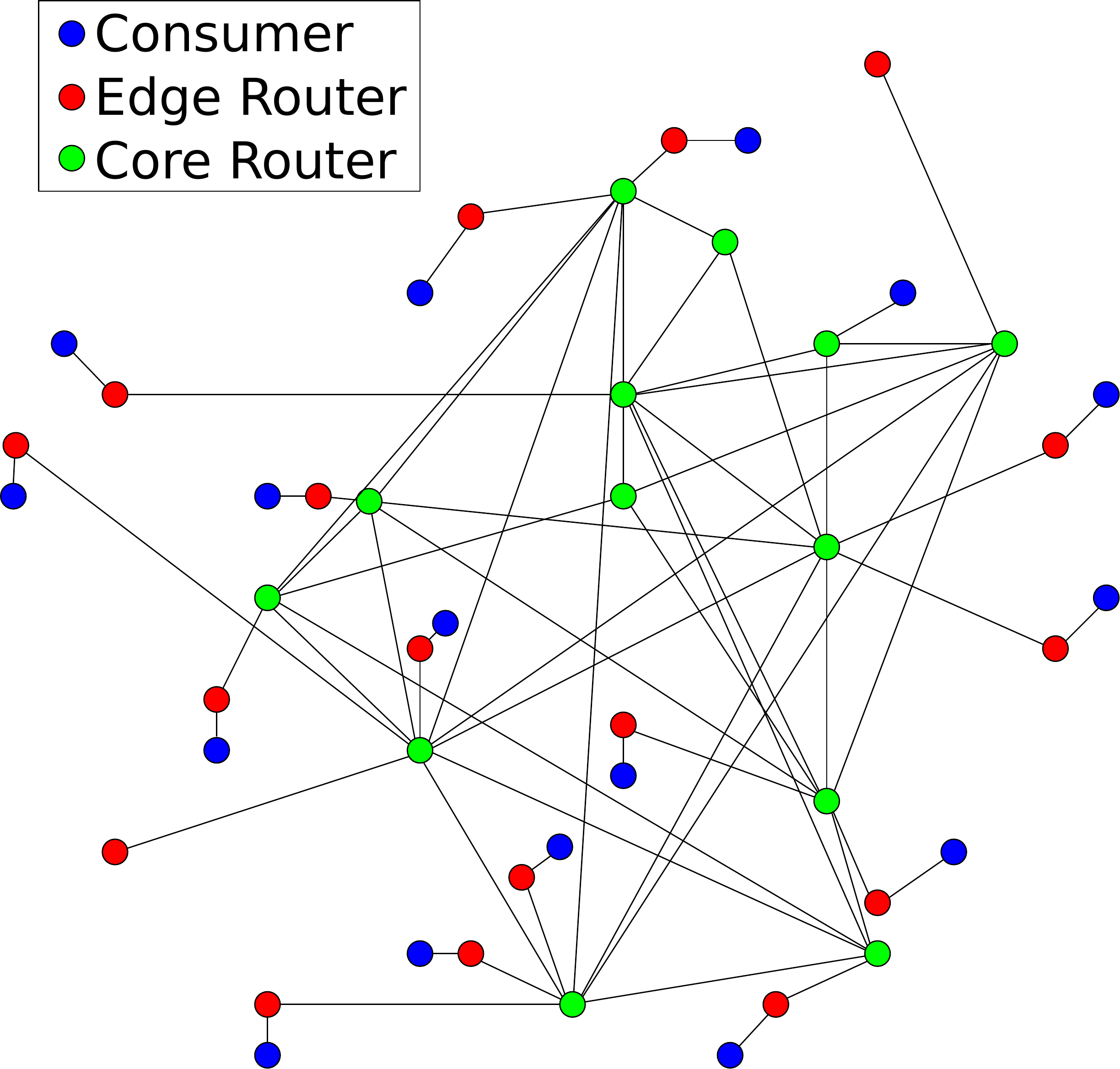}}
	\label{fig:dfn_topology}
}
\hspace{0.35cm}
\subfigure[DFN topology results with different rates of pre-populated fake content
objects (FCP: percentage of pre-populated fake content objects)]
{
	\includegraphics[height=2.3in,width=\columnwidth]{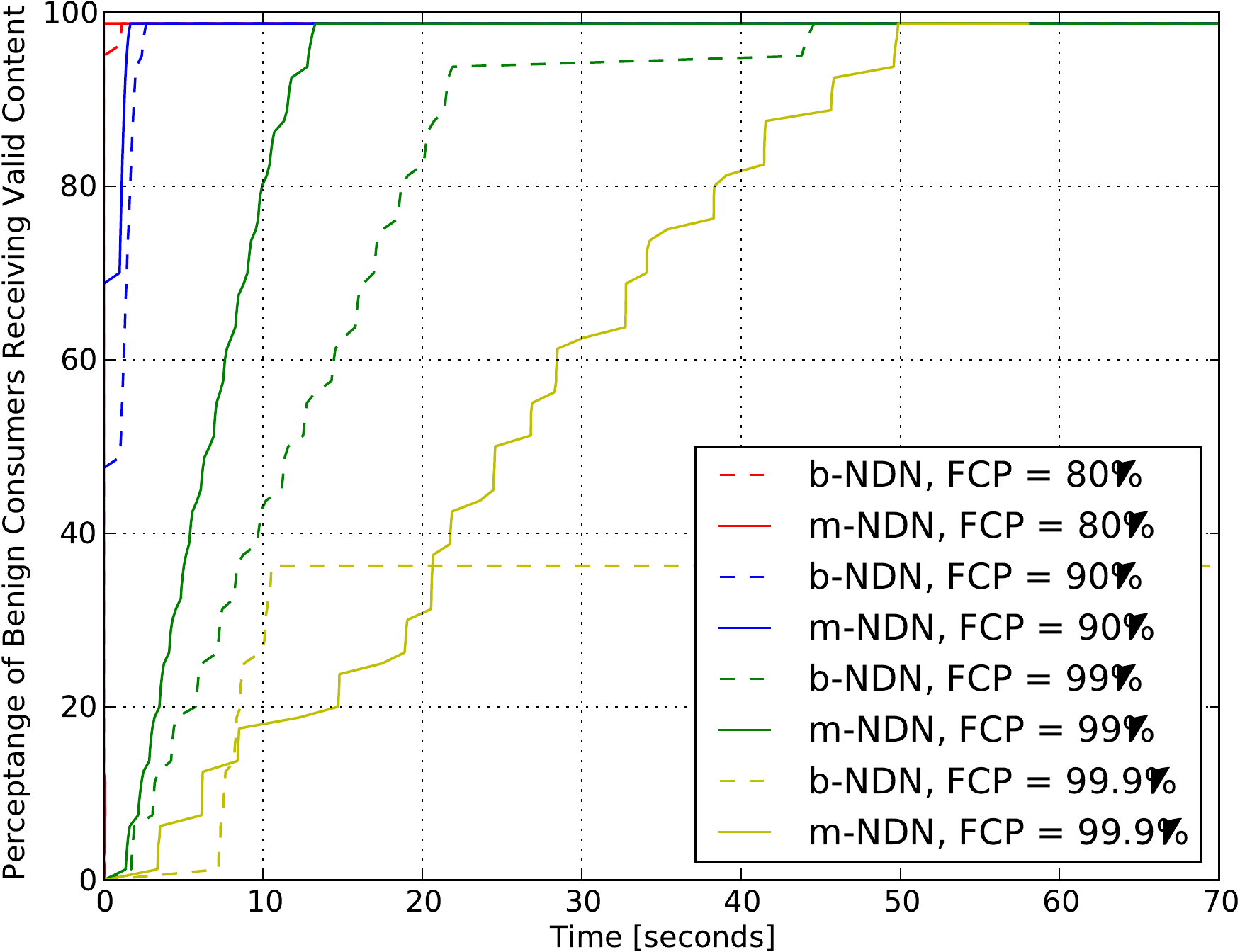} 
	\label{fig:dfn_results}
}
\caption{Content Poisoning Attack}
\label{fig:content_poisoning_attack}
\end{figure*}

\section*{APPENDIX A}
We first describe the adversary model and how fake content can be injected. We consider a 
proactive content poisoning attack whereby \Adv\ anticipates interests for content $C$ 
with name $n$ and injects fake content with the same name into router caches. Fake content 
can be injected into the network via malicious  routers or end-nodes.  For example, consider 
an \Adv\ (consisting of malicious consumer $Cr_m$ and a malicious producer $P_m$) targeting 
a specific victim router $R_v$. Assuming that $Cr_m$ and $P_m$ are connected to different interfaces 
of $R_v$, $Cr_m$ sends an interest for $n$. Once this interest is received by $R_v$ and an 
entry is added to the PIT, $P_m$ sends a fake content to $R_v$ which is promptly cached. 
Consequently, $R_v$ is pre-polluted with fake content, ready for arrival of genuine interests. 
To maximize longevity of the attack, $P_m$ sets the freshness field of fake content 
to a maximum value.

We simulate the DFN topology, 
Deutsches ForschungsNetz (German Research Network)~\cite{DFNverein, DFN-NOC} -- a 
network developed for research and education purposes. It consists of several connected 
routers positioned in different areas of the country, as shown in Figure~\ref{fig:dfn_topology}.
Our experiment measures how many benign consumers can retrieve a satisfactory (genuine) 
content and how fast they can do so when the router caches are poisoned. The simulation 
starts with core router caches pre-populated with various fake versions of the target 
content, $80\%$ (1 valid and 4 fake content objects), $90\%$ (1 valid 9 fake objects),
$99\%$ (1 valid 99 fake objects), and $99.9\%$ (1 valid 999 fake objects). To show the 
effect of having multiple consumers connected to the same router, we configure edge router 
to run without cache. Figure~\ref{fig:dfn_results} shows the results of this experiment. 
We can notice that it takes more than 20 seconds for $90\%$ of the consumers to retrieve 
valid content in the case pre-populated fake content objects rate of $99\%$. Moreover, more than $60\%$
of the consumers do not receive valid content during the time of the simulation, when the pre-population
rate is $99.9\%$.

\section*{APPENDIX B}
\begin{defn}
A hash function $\mathcal{H}$ is {\em second pre-image resistant}, if for any 
given $x$, no probabilistic polynomial-time (PPT) adversary $\mathcal{A}$ can find a value $x' \neq x$ such 
that $\mathcal{H}(x) = \mathcal{H}(x')$. In other words, 
$\Pr \left[ \mathcal{H}(x) = \mathcal{H}(x') \right] \leq \epsilon(n)$, where $\epsilon(n)$ 
is negligible and $n$ is the security parameter.  
A formal definition of probabilistic polynomial-time adversaries and negligible functions 
can be found in \cite{katz2008introduction}
\end{defn}

\begin{defn}
A signature scheme $\Pi$ is unforgeable if for any message $m$, no PPT
adversary $\mathcal{A}$ (given a public key $PK$) can generate a valid signature without 
knowing the corresponding private key. We denote the success of $\mathcal{A}$ as 
$\mathcal{A}^{\textrm{\bf forge}}(m) = 1$, i.e., if $\Pi$ is unforgeable, there exists 
a negligible function $\epsilon(n)$ such that:
$\Pr \left[ \mathcal{A}^{\textrm{\bf forge}}(m) = 1 \right] \leq \epsilon(n)$.
\end{defn}

\begin{defn}
For any interest message $\mathit{Int}$ with $\mathcal{H}(PK)$ (the digest of the verifying 
public key for the corresponding content) assigned to the \texttt{PPKD} 
field, and for any $\mathcal{A}$, the NDN cache poisoning 
experiment is defined as follows:

Given $\mathit{Int}$ as input to $\mathcal{A}$, it outputs a content object $C'$ containing: 
(1) a public key $PK'$ in the \texttt{KeyLocator} field, (2) a digest of this key $\mathcal{H}(PK')$ 
in \texttt{PPKD}, and (3) a signature $\sigma'$ in the \texttt{Signature} field. 
The output of this experiment is defined to be $1$ if one of the following holds:
\begin{compactitem}
\item $PK \neq PK'$ and $\mathcal{H}(PK) = \mathcal{H}(PK')$,
\item or, $PK = PK'$ and $\sigma$ is valid.
\end{compactitem}
In other words, $\mathcal{A}$ can either violate the second pre-image resistance of $\mathcal{H}$
(we denote this event as {\bf collision} which occurs with some probability $p_c$ and succeeds with 
probability $\Pr \left[ \mathcal{H}(x) = \mathcal{H}(x') \right]$), or forge the signature (we denote this event 
as {\bf forge} which occurs with some probability $p_f$ and succeeds with 
$\Pr \left[ \mathcal{A}^{\textrm{\bf forge}}(m) = 1 \right]$). We denote the success of $\mathcal{A}$ 
as $\mathcal{A}^{\textrm{\bf pois}}(\mathit{Int}) = 1$.
\end{defn}

\begin{theorem}
Given $\mathcal{H}$, $\Pi$ (as defined above), $\mathcal{A}$ succeeds in injecting a fake content object 
$C'$ into a network that abides by the IKB rule with a negligible probability $\epsilon(n)$.
\begin{align*}
\Pr \left[ \mathcal{A}^{\textrm{\bf pois}} \left( \mathit{Int} \right) = 1 \right] \leq \epsilon(n)
\end{align*}
\end{theorem}

\begin{proof}
We show the above by contradiction:

Assume that $\mathcal{A}$ succeeds in injecting $C'$ with a non\hyp{}negligible 
probability. Then, we can construct a reduction $\mathcal{A}'$ (another PPT
adversary), that uses $\mathcal{A}$ to break second pre-image resistance of $\mathcal{H}$, or  
unforgeability of $\Pi$:

\vspace{0.2cm}
\centerline{\fbox{\parbox{0.968\columnwidth}{
\textbf{Adversary $\mathcal{A}'$}
\begin{compactenum}
\item Is given a hash value $x$.
\item Creates an interest message $\mathit{Int}$ and sets $\mathcal{H}(x)$ as its 
\texttt{PPKD} field value.
\item Runs $\mathcal{A}(\mathit{Int})$ to obtain $C'$.
\item Extracts from $C'$ and outputs:
\begin{compactenum}
\item $PK'$ as a collision with $x$, if $x \neq PK'$,
\item or $\sigma'$ as a forged signature for $C'$, if $x = PK'$.
\end{compactenum}
\end{compactenum}
}}}

We now determine the probability of success of $\mathcal{A}'$. Whenever either {\bf collision} or 
{\bf forge} event occurs $\mathcal{A'}$ succeeds. Therefore,
\begin{align*}
\label{equ:security_proof_1}
\Pr \left[ \mathcal{A'}~\mathrm{succeeds} \right] &= \Pr \left[ {\bf collision} \cup {\bf forge} \right] \nonumber\\
&= p_c \cdot \Pr \left[ \mathcal{H}(x) = \mathcal{H}(PK') \right] \nonumber\\
&+ p_f \cdot \Pr \left[ \mathcal{A'}^{\mathrm{\bf forge}}(C') = 1 \right] \nonumber\\
&> \epsilon(n)
\end{align*}
The last inequality holds because $\mathcal{A}'$ succeeds with the same probability as $\mathcal{A}$, 
which is non-negligible. If the result of adding two functions is non-negligible, at least one of them must be 
non-negligible~\cite{katz2008introduction}. Moreover, since both $p_c$ and $p_f$ cannot be 
exponential functions, then either\\
$\Pr \left[ \mathcal{H}(x) = \mathcal{H}(PK') \right] > \epsilon(n)$ or 
$\Pr \left[ \mathcal{A'}^{\mathrm{\bf forge}}(C') = 1 \right] > \epsilon(n)$.

\end{proof}

\end{document}